\documentclass[journal]{IEEEtran}

\usepackage{lineno}
\usepackage{hyperref}
\usepackage{amssymb}
\usepackage{epsfig}
\usepackage{graphicx}
\usepackage{amsmath}
\usepackage{amsthm}
\usepackage{amsbsy}
\usepackage{array}
\usepackage{algorithm} 
\usepackage{algorithmic}
\usepackage{subfigure}
\usepackage{caption}
\usepackage{amsfonts}
\usepackage{amsmath}
\usepackage{mathrsfs}
\usepackage{mathtools}  
\usepackage{amssymb}
\usepackage{tabulary}
\usepackage{booktabs}
\usepackage{array}
\usepackage{makeidx}                          
\usepackage{multicol}        
\usepackage[bottom]{footmisc}

\DeclareMathOperator*{\argmax}{argmax}
\newtheorem{theorem}{Theorem}
\newtheorem{definition}{Definition}

\allowdisplaybreaks[1]

\ifCLASSINFOpdf
\else
\fi

\begin{document}
\title{A Generalised Differential Framework for Measuring Signal Sparsity}

\author{Anastasios~Maronidis, Elisavet~Chatzilari, Spiros~Nikolopoulos and Ioannis~Kompatsiaris}
\maketitle

\begin{abstract}
The notion of signal sparsity has been gaining increasing interest in information theory and signal processing communities. 
Recent advances in fields like signal compression, sampling and analysis have accentuated the crucial role of sparse representations of signals. 
As a consequence, there is a strong need to measure sparsity and towards this end, a plethora of metrics has been presented in the literature.
The appropriateness of these metrics is typically evaluated against a set of 
objective criteria that has been proposed for assessing the credibility of any sparsity metric.
In this paper, we propose a Generalised Differential Sparsity (GDS) framework for generating novel sparsity metrics whose functionality is based on the concept that sparsity is encoded in the differences among the signal coefficients. 
We rigorously prove that every metric generated using GDS satisfies all the aforementioned criteria and we provide a computationally efficient formula that makes GDS suitable for high-dimensional signals. 
The great advantage of GDS is its flexibility to offer sparsity metrics that can be well-tailored to certain requirements stemming from the nature of the data and the problem to be solved.
This is in contrast to current state-of-the-art sparsity metrics like Gini Index (GI),
which is actually proven to be only a specific instance of GDS, demonstrating the generalisation power of our framework. 
In verifying our claims, we have incorporated GDS in a stochastic signal recovery algorithm and experimentally investigated its efficacy in reconstructing randomly projected sparse signals. 
As a result, it is proven that GDS, in comparison to GI, both loosens the bounds of the assumed sparsity of the original signals and reduces the minimum number of projected dimensions, required to guarantee an almost perfect reconstruction of heavily compressed signals.
The superiority of GDS over GI in conjunction with the fact that the latter is considered as a standard in numerous scientific domains, prove the great potential of GDS as a general purpose framework for measuring sparsity.
\end{abstract}


\IEEEpeerreviewmaketitle

\section{Introduction}
\label{Intro}

\IEEEPARstart{S}{parse} representation of signals has been celebrated as a premise that permits the solution of problems previously unsolvable, paving the way to unprecedented possibilities in fields like signal compression and reconstruction.
Roughly speaking, sparsity measures the extent to which the information of a signal is distributed to the coefficients.
More specifically, for highly sparse signals the information is concentrated to a small portion of coefficients, while for non-sparse signals the information is uniformly distributed across the coefficients.
In this context, sparsity is a desirable property as it allows for succinct representations of large pieces of information. 
Recall the Occam's razor, which dictates that among a set of representations, the most compact is always preferred \cite{hamilton1855discussions}.

There are many paradigms stemming from diverse research domains advocating the importance of sparsity.
Compressive Sampling (CS) \cite{candes2006compressive} comprises the most vivid example, where the role of sparsity has been demonstrated in the process of compressing and reconstructing a signal. 
More specifically, through the introduction of the Null Space Property (NSP) \cite{baraniuk2011introduction} and the Restricted Isometry Property (RIP) \cite{cande2008introduction}, it has been proven that under the assumption of data sparsity, it is possible to solve an underdetermined linear system of equations. 
The above important result allows for the perfect reconstruction of  a signal that has been compressed using only few random projections of the original sparse signal. 
Towards this end, a variety of optimisation algorithms, which incorporate the notion of sparsity, has been proposed for reconstructing a compressed signal.
For instance, the Dantzig selector solves an $l_1$-regularisation problem in an attempt to estimate a ground truth sparse signal from few noisy projections of this signal \cite{candes2007dantzig}. 
In a similar vein, sparsity has also been utilised in the Lasso algorithm for recovering sparse representations of high-dimensional signals \cite{meinshausen2009lasso}. 

Apart from the aforementioned applications, the notion of sparsity has also been incorporated in already existing methods in various fields. 
For instance, Bayesian methods providing sparse solutions to regression and classification problems have attracted a renewed interest \cite{tipping2001sparse}. 
Moreover, in Support Vector Machines (SVM), optimal guarantees on the sparsity of the support vector set encoding the boundary between two classes, have also been investigated \cite{cotter2013learning}.
Sparsity appears to play a key role in boosting techniques as well, leading to sparse combinations of a number of weak classifiers \cite{xi2009speed}.  
Additionally, in an unsupervised configuration, Sparse Principal Component Analysis (S-PCA) has been introduced as a framework, which trades off redundancy minimisation for sparsity maximisation in the basis signals \cite{chennubhotla2001sparse}. 

The above comprise only indicative examples from an endless catalogue of diverse scientific domains, where signal sparsity finds application. 
Nevertheless, they are fairly enough to demonstrate that the notion of sparsity occupies a dominant position in signal processing and information theory communities. 
Given the importance of sparsity, it is essential to find an effective way to measure it.
Apparently, the way sparsity is defined and measured is dictated by the specific purpose it is designed to serve. 
In this paper, in the context of CS, we are particularly concerned with the role of sparsity in the reconstruction of signals which have been heavily compressed using random projections.
Signal reconstruction covers a large portion of the problems that concern sparsity. 
Hence, the conclusions drawn from our analysis are expected to have impact on other case studies as well.

Formally, the core idea of sparsity, as this has originally been introduced in CS, is to count the integer number of non-zero coefficients of a signal, measured with the help of the $l_0$ norm \cite{candes2006compressive}.
In practice though, this proves to be a very strict definition, as rarely in real-world problems signals contain exact zeros. 
As a consequence, the research community has resorted to new relaxed measures of sparsity whose actual objective is to estimate an approximation of the number of non-zero coefficients, allowing sparsity to take decimal values.
Along these lines, the notion of sparsity is usually referred to as signal compressibility instead.
From now on though, in our work, we will consistently use the term sparsity even in the cases where we will actually refer to signal compressibility. 

There has been proposed a variety of sparsity metrics in the literature \cite{hurley2009comparing}.
Among them, Gini Index (GI) \cite{hurley2009comparing,gini1921measurement} offers a state-of-the-art solution, which has led to impressive results in reconstructing compressed signals \cite{zonoobi2011gini}.
The validity of the majority of sparsity metrics often relies merely on intuitive criteria. 
To overcome this drawback, in \cite{hurley2009comparing, pastor2013sparsity}, a number of objective criteria has been proposed. 
The origin of these criteria stems from the financial science, where the notion of sparsity is analogous to the inequity of wealth distribution in a human society \cite{dalton1920measurement}. 
The aforementioned criteria provide a degree of credibility to a sparsity metric enabling the comparison between different metrics.

Most of the already existing sparsity metrics, e.g., $l_1$ norm, use the magnitude of the vector coefficients to encode sparsity. 
A drawback though is that in this way the relativity among the coefficients is completely defied.
In this paper, we claim that in the process of measuring the sparsity of a vector, it is not the absolute value of the coefficients which is important, but their relative differences. 
The rationale of our claim relies upon the observation that how small or large a value is, depends on what reference value it is compared to. 
Therefore, by calculating the relative differences, we actually compare the coefficients to each other.
The following example elaborates on why relativity of coefficients might prove to be crucial.

Considering $\mathbf{c} = [10^{-10}, 10^{-10}, 10^{-10}, 10^{-10}, 10^{-10}]$,
two questions naturally emerge.
How close are the coefficients of $\mathbf{c}$ to zero and thereby how sparse is $\mathbf{c}$?
The answer to the first question is that they are all equally close to zero, regardless of what is their actual distance from it.
In this sense, all coefficients contain equal percentage of the signal energy, which means that all of them are equally important in the representation of the signal and hence none of them should be discarded. 
As a consequence, although at a first glance it may look counter-intuitive, the answer to the second question is clearly that $\mathbf{c}$ is totally non-sparse.

Motivated by the previous analysis, as the main contribution of this paper, we propose a Generalised Differential Sparsity (GDS) framework, whose functionality is based on the differences among the signal coefficients. 
The advantageous feature of our framework is that it is customisable to certain data types and problem requirements, due to an adjustable parameter, which from now on we call the order of GDS. 
Different values of the order generate novel metric-instances with varying strictness in estimating signal sparsity.
As part of our analysis, we rigorously prove that these GDS metric-instances satisfy all the objective criteria for sparsity metrics \cite{hurley2009comparing, pastor2013sparsity}.
Moreover, we prove that GDS of first order collapses to GI. 
The encapsulation of GI within GDS emphasises the generalisation power of the latter in unifying already existing metrics apart from generating novel ones.
In addition, although the computation of GDS using its original formula is tractable even for large values of its order, it proves to be cumbersome for high-dimensional data. 
For dealing with the above shortcoming, we provide an equivalent formula of GDS, which allows for its efficient calculation when the number of dimensions is high.
The drawback though of the latter formula is that in contrast to the original one, it is costly for big values of the order of GDS. 
Consequently, both formulas prove to be useful and can be used interchangeably according to the given circumstances.

The order of GDS determines the tendency of the corresponding metric-instance to qualify an arbitrary signal as sparse. 
This proves to be a great advantage, since it offers GDS the flexibility to adjust to certain requirements arising from the nature of the data and the problem to be solved. 
In order to verify the above claim, we have used GDS to reconstruct 
sparse signals which have been heavily compressed via random projections. 
For this purpose we have employed the reconstruction approach presented in \cite{zonoobi2011gini}, which combined with GI has returned excellent results.
The reconstruction is performed by incorporating a sparsity metric into a stochastic approximation method that solves a dedicated sparsity maximisation problem. 
More specifically, given a compressed signal and based on the prior assumption that the original signal before compression was sparse, the idea is to find in the original space, the signal with highest sparsity that gives the smallest reconstruction error. 

Through an experimental study similar to the one presented in\cite{zonoobi2011gini}, we prove that incorporating GDS to the previous reconstruction approach, in comparison with the top performing GI, loosens the assumptions of both the underlying sparsity of the original signal and the required number of projected dimensions.
In other words, GDS offers further compression capacity to lowly sparse signals and simultaneously allows for using a smaller number of projected dimensions without increasing the reconstruction error.
Along the same lines, it is proven that the optimal order of GDS is strongly dependent on the type and sparsity of the original data as well as the desired compression level. 
This finding justifies the rationale behind using different values for the order of GDS and provides a useful rule of thumb in deciding what order of GDS is the appropriate for certain problem parameters.

The remainder of this paper is organised as follows: In Section \ref{RelatedWork}, a number of works related to measuring signal sparsity is reviewed. 
In Section \ref{Properties}, we present a set of desirable criteria a sparsity metric must obey and we give intuitive interpretations of them. 
In Section \ref{GDSp}, we propose the novel GDS framework and we rigorously prove that it satisfies all the aforementioned criteria. 
In Section \ref{Gini}, we prove that GI is encapsulated within our generalised framework. 
In Section \ref{Efficient}, we provide a computationally more efficient formula for calculating the GDS of a signal. 
A normalised version of GDS based on the statistics of the data is also presented in Section \ref{NormalGDS}.
Subsequently, through a series of experiments, in Section \ref{Experiments}, we study GDS signal reconstruction performance as a function of the type and the sparsity of the original signals, the number of projected dimensions and the GDS order.
We conclude this paper in Section \ref{Conclusion}. Finally, an appendix containing the most lengthy mathematical proofs is provided at the end of this paper.

\section{Related Work}
\label{RelatedWork}

A variety of methods has already been proposed in the bibliography for measuring the inherent sparsity of a signal. 
The most straightforward way to measure sparsity is by using the $l_0$-norm \cite{karvanen2003measuring}. 
Although it has led to impressive theoretical results in sparse representation \cite{opac-b1133206}, in practice the $l_0$-norm suffers from several disadvantages. 
For instance, changing the non-zero coefficients does not affect it, while even an infinitesimal amount of noise on the zero coefficients may dramatically distort it.
For overcoming such disadvantages, approximations of the $l_0$-norm have been proposed in sparsity optimisation problems in the presence of noise \cite{fuchs2005recovery,donoho2006stable}. 
Thresholding techniques have also been employed \cite{rath2008sparse}, however the selection of a reasonable threshold may prove to be problematic. 

Due to its disadvantages, the $l_0$-norm has often been replaced by the $l_1$-norm, which offers a plausible alternative metric surpassing some of the shortcomings accompanying the former \cite{candes2005decoding,donoho2008fast}. 
Towards this direction, \cite{candes2005decoding} comprises a milestone work, where the authors prove that the classical error correcting problem, under certain conditions can be translated into an $l_1$-optimisation problem. 
The latter can be trivially solved in a linear programming configuration using existing dedicated methods. 
In a similar vein, the authors of \cite{donoho2008fast} employ the Homotopy method to solve an underdetermined system of linear equations through an $l_1$-minimisation problem. 
In \cite{candes2008enhancing}, the authors propose a methodology for sparse signal recovery that often outperforms the $l_1$-minimisation problem by reducing the number of measurements required for perfect reconstruction of the compressed signal. 
The problem is decomposed into a sequence of $l_1$-minimisation sub-problems, where the weights are updated at each iteration based on the previous solution. 

Norms $l_p$ of higher order have also been used for measuring signal sparsity \cite{karvanen2003measuring}. 
Moreover, combinations of $l_p$ norms have been proposed as well. 
For instance, the Hoyer sparsity metric based on the relationship between the $l_1$ and the $l_2$ norm has been utilised in a sparsity constrained Non-negative Matrix Factorisation (NMF) setting for finding linear representations of non-negative data \cite{hoyer2004non}. 
Apart from the $l_p$ norms, other mathematical functions have also been used for measuring signal sparsity. 
For example, kurtosis has been proposed for data following a unimodal and symmetric distribution form \cite{olshausen2004sparse}. 
In the same vein, the authors or \cite{karvanen2003measuring} suggest the adoption of $tanh$ functions as an approximate solution of $l_p$ norms. 
Furthermore, they introduce a metric based on order statistics. 
In contrast to $l_p$ norms and similarly to our work, the functionality of such methods is based on the distribution form of the signal coefficients rather than their magnitudes.
A main drawback though is that they can only handle signals whose coefficients contain a unique dominant mode at zero, and thus must be avoided when dealing with signals containing multiple modes, which constrains their scope of applications \cite{karvanen2003measuring}. 

The connection between sparsity and entropy has been clearly demonstrated in \cite{pastor2013sparsity}.
Entropy expresses the complexity of a signal, while sparsity expresses its compressibility under appropriate basis.
Along these lines, the authors argue that both sparsity and entropy should follow similar intuitive criteria. 
Towards this end, they propose a novel sparsity and a novel entropy metric that satisfy such criteria.
The functionality of these metrics is based on the calculation of the similarity of a signal to the theoretically totally non-sparse one using the inner product between them.
Relying on the above connection, entropy diversity metrics can also be used to measure sparsity \cite{rao1999affine,kreutz1998measures}. 
For instance, the Shannon and the Gaussian entropy presented in \cite{rao1999affine} constitute plausible measures of sparsity, which incorporated in sparsity minimisation problems may lead to sparse solutions to the best basis selection problem \cite{kreutz1998measures}.

Among the most prevalent sparsity metrics, the Gini Index (GI) \cite{hurley2009comparing} offers a state-of-the-art solution in a variety of applications.
As a couple of examples, it has been effectively used for maximising the sparsity of wavelet representations via parameterised lifting \cite{hurley2007maximizing} as well as for finding the most sparse representation of speech signals \cite{rickard2004gini}.
Moreover, in relation to our work, GI has shown top performance in recovering randomly projected signals \cite{zonoobi2011gini}.
Therefore, a comparison with it is mandatory to prove the potential of our framework.
In this context, it is worth noticing that both GI and our framework build on common incentives.
More specifically, for both methods, what is important is the relativity among coefficients rather than their absolute magnitudes.
Actually as we will see later in this paper, our framework shares all the advantages accompanying GI plus the extra advantage that contains an adjustable parameter, which makes our method a sparsity metric generator instead of a single sparsity metric.

Most of the sparsity metrics mentioned in this section along with others, like for instance the \emph{log} measure or the $u_{\theta}$, have been collectively reviewed and compared in \cite{hurley2009comparing}. 
Although each of them has its own advantages, a number of objective criteria has been proposed in the literature for assessing their performance. 
This set of criteria serves as a benchmark that allows for comparing different metrics to each other. 
In \cite{hurley2009comparing}, the authors summarise which of the sparsity criteria are satisfied by each of the sparsity metrics proposed in the literature. 
Interestingly, the only metric that satisfies all criteria is GI. 
In this paper we will prove that our proposed framework does as well. 
In the following section, we provide the full set of the above objective criteria, acquired from the literature \cite{hurley2009comparing, pastor2013sparsity, dalton1920measurement}.

\section{Sparsity metric objective criteria}
\label{Properties}

For the remainder of this paper, we will borrow the term vector from linear algebra in order to refer to a signal. 
So, let
\begin{equation*}
\mathbf{c} = \left[ c_1,\dots,c_N \right] \in \mathbb{R}^N 
\end{equation*}
be an $N$-length vector whose sparsity we would like to measure. A sparsity metric is a function $S: \mathbb{R}^N \rightarrow \mathbb{R}$, which given 
$\mathbf{c}$
returns a real number $S(\mathbf{c})$ that comprises an estimation of its sparsity. From now onwards, we implicitly assume that sparsity is measured using the magnitudes of the coefficients and not their algebraic values, i.e., the coefficient signs can be neglected. Therefore, for simplicity, we can assume that 
$$c_i \ge 0, \, \forall \, i \in \{ 1, 2, \dots, N \}.$$

In the following, we give the mathematically rigorous definitions of the above-mentioned sparsity criteria along with intuitive interpretations. 

\begin{itemize}

\item \emph{$P_1$: Continuity}
\begin{equation}
S(\mathbf{c} + d \mathbf{c}) \rightarrow S(\mathbf{c}), \,\, \text{when} \,\, d \mathbf{c} \rightarrow \mathbf{0}.
\end{equation}
This property requires that small changes of the coefficients should not lead to dramatic change of sparsity.

\item \emph{$P_2$: Permutation Invariance}
\begin{equation}
S(t(\mathbf{c})) = S(\mathbf{c}), \,\, \text{where} \,\, t(\mathbf{c}) \,\, \text{is a permutation of} \,\, \mathbf{c}.
\end{equation}
This property postulates that permuting the coefficients of a vector should not affect its sparsity. 
Provided that this property holds for a metric -- which is usually the case -- for convenience and without loss of generality, given an arbitrary vector, since the position of the coefficients does not matter, 
we can consider that these have been sorted in ascending order, i.e.: 
\begin{equation}
0 \le c_1 \le c_2 \le \dots \le c_N.
\end{equation}

\item \emph{$P_3$: Robin Hood}

Let 
$$\mathbf{c} = \left[ c_1, c_2, \dots, c_N \right].$$
If 
$$\mathbf{c}' = \left[ c_1,\dots,c_i+a,\dots,c_j-a,\dots,c_N \right],$$
then 
\begin{equation*}
S(\mathbf{c}') < S(\mathbf{c}), \text{\, for all \,} a,c_i,c_j, \text{\, such that \,} 
\end{equation*}
\begin{equation*}
c_j>c_i \text{\, and \,} 0<a<\frac{c_j-c_i}{2}.
\end{equation*}
Robin Hood property says that subtracting a specific amount from a large coefficient and adding this amount to a smaller coefficient decreases the vector-sparsity as the energy of the vector spreads out along the coefficients.
The constraint of $a$ is used to avoid $c'_i>c'_j$.

\item \emph{$P_4$: Scaling}
\begin{equation*}
S(a\mathbf{c}) = S(\mathbf{c}), \text{\,\,} \forall a \in \mathbb{R}, \text{\,\,} a>0.
\end{equation*}
Scaling property requires that by multiplying all vector coefficients with the same scalar must not affect vector-sparsity. 

\item \emph{$P_5$: Rising Tide}
\begin{equation*}
S(\mathbf{c} + a) < S(\mathbf{c}), \text{\,\,} a \in \mathbb{R}, \text{\,\,} a>0, 
\end{equation*}
\begin{equation*}
\text{except for the case where} \text{\,\,} c_1 = c_2 = \dots = c_N.
\end{equation*}
Rising Tide says that adding the same scalar to all vector coefficients reduces vector-sparsity. The significance of this property becomes more obvious by examining the limit behaviour of the vector under this operation. Indeed, by adding an increasing amount to all coefficients, 
the relative difference among coefficients becomes negligible and
therefore the sparsity should
asymptotically become zero.

\item \emph{$P_6$: Cloning} 
\begin{equation*}
S(\mathbf{c}) = S(\mathbf{c} \| \mathbf{c}) = \dots = S(\mathbf{c} \| \mathbf{c} \| \dots \| \mathbf{c}), 
\end{equation*}
\begin{equation*}
\text{\, where \,} \| \text{\, denotes concatenation. \,}
\end{equation*}
Cloning requires that concatenating a number of vectors, which comprise exact copies of the original one must not affect vector-sparsity. This is also quite reasonable, if we take again into account that the relative difference of the vector coefficients after sorting the resulting vector is kept intact. 

\item \emph{$P_7$: Bill Gates}
\small
\begin{multline*}
\forall i \in \{ 1, 2, \dots, N \}, \text{\,\,} \text{and} \text{\,\,} \forall a>0: 
\\
S(\left[ c_1,\dots,c_i+a,\dots,c_n\right])
>S(\left[c_1,\dots,c_i,\dots,c_n\right]).
\end{multline*}
\normalsize
By increasing the value of a vector coefficient, while maintaining the remaining coefficients, the sparsity increases, as the vector energy is concentrated to a mere coefficient.

\item \emph{$P_8$: Babies}
\begin{equation*}
S(\mathbf{c} \| 0)>S(\mathbf{c}).
\end{equation*}
By adding extra zero's to the original vector, the vector-sparsity increases. This action has similar effect to $P_7$, since by adding zero's the energy is concentrated to fewer coefficients.

\item \emph{$P_9$: Saturation}
\begin{equation}
\lim_{N \rightarrow + \infty} \frac{S(\mathbf{0}_N || 1)}{S(\mathbf{0}_{N-1} || 1)} = 1.
\end{equation}
Saturation says that by concatenating extra zeros to a vector, the change of its sparsity asymptotically becomes negligible.

\item \emph{$P_{10}$: Lower Bound}
\begin{equation}
S(\mathbf{c}) \ge S(\mathbf{1}_N).
\end{equation}
The smallest possible sparsity is encoded to a vector consisting of ones.

\item \emph{$P_{11}$: Upper Bound}
\begin{equation}
S(\mathbf{c}) \le S(\mathbf{0}_{N-1} || 1).
\end{equation}
The largest possible sparsity is encoded to a vector consisting of all but one zeros.

\end{itemize}

\section{Generalised Differential Sparsity}
\label{GDSp}

As already mentioned in Section \ref{RelatedWork}, the central idea of sparsity is based on the number of zero coefficients of a vector. As a consequence, most of the already existing sparsity metrics (e.g., the $l_p$-norm \cite{karvanen2003measuring}), use the magnitude of the vector coefficients to encode sparsity. 
In contrast to the above, in this section we propose a novel Generalised Differential Sparsity (GDS) metric framework, which takes into account the differences among the coefficients of a vector, rather than the magnitudes per se, for measuring its sparsity.
In this way, GDS achieves to measure the sparsity of a vector by examining the extent to which the energy is distributed to the coefficients.

\begin{definition}
\label{GDS_def}
The GDS of order $p$ ($p \ge 1$) of a non-zero vector $\mathbf{c} \in \mathbb{R}^N$ is defined as:
\begin{equation}
\label{GDS_definition}
S_p(\mathbf{c}) = \frac{1}{N \sum_{i=1}^N c_i^p} \sum_{i=1}^{N-1} \sum_{j=i+1}^{N} \left( c_j - c_i \right)^p, 
\end{equation}
where the coefficients have been sorted in ascending order so that:
$c_1 \le c_2 \le \dots \le c_N$.
\end{definition}

Using Definition \ref{GDS_def}, we prove the following theorem, which provides lower and upper bounds on the possible values of GDS regardless of its order $p$.

\begin{theorem} 
\label{Range}
$0 \le S_p(\mathbf{c}) \le 1 - \frac{1}{N}, \, \, \, \forall \, \mathbf{c} \in \mathbb{R}^N, \, \, \, \forall \, p \ge 1$.
\end{theorem}

\begin{proof}
\begin{equation}
\label{GDS}
S_p(\mathbf{c}) = \alpha \sum_{i=1}^{N-1} \sum_{j=i+1}^N \left( c_j - c_i \right)^p,
\end{equation} 
where
\begin{equation}
\alpha = \frac{1}{N \sum_{i=1}^N c_i^p}.
\end{equation}
But $\left( c_j - c_i \right)^p \le c_j^p - c_i^p$, since $c_j \ge c_i \ge 0$ (recall that the coefficients are sorted).
Therefore, (\ref{GDS}) becomes:
\begin{multline*} 
S_p(\mathbf{c}) \le \alpha \sum_{i=1}^{N-1} \sum_{j=i+1}^N (c_j^p - c_i^p )
\\ 
= \alpha \left[ \sum_{i=1}^{N-1} \sum_{j=i+1}^N c_j^p - \sum_{i=1}^{N-1} \sum_{j=i+1}^N c_i^p \right]
\\
= \alpha \Big[ ( c_2^p + c_3^p + \dots + c_N^p ) + ( c_3^p + \dots + c_N^p ) + \dots + c_N^p 
\\
- (N-1) c_1^p - (N-2) c_2^p - \dots - 2c_{N-2}^p - c_{N-1}^p \Big] 
\\
\le \alpha \Big[ ( c_2^p + c_3^p + \dots + c_N^p ) + ( c_3^p + \dots + c_N^p ) + \dots + c_N^p \Big] 
\\
= \alpha \Big[ c_2^p + 2 c_3^p + \dots + (N-1) c_N^p \Big]  
\\
\le \alpha \Big[ (N-1) c_1^p + (N-1) c_2^p + \dots + (N-1) c_N^p \Big] 
\\
= \alpha (N-1) \sum_{i=1}^N c_i^p = \frac{N-1}{N} = 1 - \frac{1}{N}. 
\end{multline*}
Finally, it is obvious that $S_p(\mathbf{c}) \ge 0$, since in the sum always $j > i$ and therefore $c_j - c_i \ge 0$.
Hence, $0 \le S_p(\mathbf{c}) \le 1 - \frac{1}{N}$. 
\end{proof}

\noindent
Furthermore, it can be easily shown that $S_p \left( \left[ q, q, \dots, q \right] \right) = 0$ and 
$S_p \left( \left[ 0, 0, \dots, q \right] \right) = 1 - \frac{1}{N}$, which means that $0$ and $1 - \frac{1}{N}$ is the minimum and maximum sparsity value, respectively.

In the following, we provide rigorous proofs that the proposed GDS metric satisfies the eleven basic properties presented in Section \ref{Properties}. The \emph{Continuity} and \emph{Permutation Invariance} properties can be trivially proven from the definition, while the proof of \emph{Lower} and \emph{Upper Bound} properties is contained in Theorem \ref{Range}. 
The proof of $P_3$, due to its extensive length, has been appended to the end of this paper. Hereunder, we provide the proofs of properties $P_4$ to $P_9$.

\begin{theorem}
\label{ScaleTheorem}
GDS satisfies $P_4$: Scaling, i.e.:
\begin{equation*}
S_p(a \mathbf{c}) = S_p(\mathbf{c}), \text{\,\,} \forall a \in \mathbb{R}, \text{\,\,} a>0.
\end{equation*} 
\end{theorem}

\begin{proof}
\begin{multline*}
S_p(a \mathbf{c}) = \frac{1}{N \left( \sum_{i=1}^N a c_i \right)^p} \sum_{i=1}^{N-1} \sum_{j=i+1}^N (a c_j -a c_i)^p
\\ 
= \frac{1}{a^p N \left( \sum_{i=1}^N c_i \right)^p} a^p \sum_{i=1}^{N-1} \sum_{j=i+1}^N (c_j - c_i)^p = S_p(\mathbf{c}).
\end{multline*}
\end{proof}

\begin{theorem}
\label{ScalingTheorem}
GDS satisfies $P_5$: Rising Tide, i.e.:
\begin{equation*}
S_p(\mathbf{c} + a) < S_p(\mathbf{c}), \,\, \forall \, a \in \mathbb{R}, \,\, a>0,
\end{equation*}
\begin{equation*}
\text{except for the case where} \text{\,\,} c_1 = c_2 = \dots = c_N.
\end{equation*}
\end{theorem}

\begin{proof}

First of all, if $c_1 = c_2 = \dots = c_N$, then clearly 
\begin{equation}
S_p(a + \mathbf{c}) = S_p(\mathbf{c}) = 0.
\end{equation}
Otherwise,
\begin{multline*}
 S_p(\mathbf{c} + a) = 
 \\
 \frac{1}{N \sum_{i=1}^N (c_i + a)^p} \sum_{i=1}^{N-1} \sum_{j=i+1}^N \left[ (c_j + a) - (c_i + a) \right]^p =
\\ 
\frac{1}{N \sum_{i=1}^N (c_i + a)^p} \sum_{i=1}^{N-1} \sum_{j=i+1}^N (c_j - c_i)^p <
\\
\frac{1}{N \sum_{i=1}^N c_i^p} \sum_{i=1}^{N-1} \sum_{j=i+1}^N (c_j - c_i)^p = S_p(\mathbf{c}),
\end{multline*}
where the inequality follows from
\begin{equation*}
\sum_{i=1}^N (c_i + a)^p > \sum_{i=1}^N c_i^p, 
\end{equation*}
therefore $S_p(\mathbf{c} + a) < S_p(\mathbf{c})$.
\end{proof}

\begin{theorem}
\label{CloningTheorem}
GDS satisfies $P_6$: Cloning, i.e.:
\begin{equation*}
S_p(\mathbf{c}) = S_p(\mathbf{c} \| \mathbf{c}) = \dots = S_p(\mathbf{c} \| \mathbf{c} \| \dots \| \mathbf{c}). 
\end{equation*}
\end{theorem}

\begin{proof}
\begin{multline*}
S_p \left( \overbrace{\mathbf{c} \| \mathbf{c} \| \dots \| \mathbf{c}}^{M-times} \right) =
\\ 
S_p \left( \Big[ \overbrace{c_1, \dots, c_1}^M, \overbrace{c_2, \dots, c_2}^M, \dots, \overbrace{c_N, \dots, c_N}^M \Big]  \right)
\\
= \frac{1}{MN \sum_{i=1}^N M c_i^p} \sum_{i=1}^{N-1} \sum_{j=i+1}^N M^2 (c_j - c_i)^p
\\ 
= \frac{1}{N \sum_{i=1}^N c_i^p} \sum_{i=1}^{N-1} \sum_{j=i+1}^N (c_j - c_i)^p = S_p(\mathbf{c}).
\end{multline*}
\end{proof}

\begin{theorem}
GDS satisfies $P_7$: Bill Gates and $P_8$: Babies, i.e.:
\begin{multline*}
\forall i \in \{ 1, 2, \dots, N \}, \text{\,\,} \exists \beta_i>0, \text{\,\,} \text{such that} \text{\,\,} \forall a>0: 
\\
S_p(\left[ c_1,\dots,c_i+\beta_i+a,\dots,c_n\right]) \!\!
> \!\! S_p(\left[c_1,\dots,c_i+\beta_i,\dots,c_n\right])
\end{multline*}
and
\begin{equation*}
S_p(\mathbf{c} \| 0)>S_p(\mathbf{c}).
\end{equation*}
\end{theorem}

\begin{proof}
The proof is straightforward by combining Theorems \ref{ScalingTheorem}, \ref{CloningTheorem} and \ref{RobinHoodTheorem} (cf. Appendix) with Theorems 2.1 and 2.2 presented in \cite{hurley2009comparing}, the former of which states that if Robin Hood and Scaling are satisfied, then Bill Gates is also satisfied and the latter of which states that if Robin Hood, Scaling and Cloning are satisfied, then Babies is also satisfied.
\end{proof}

\begin{theorem}
GDS satisfies $P_9$: Saturation, i.e.:
\begin{equation}
\lim_{N \rightarrow + \infty} \frac{S_p(\mathbf{0}_N || 1)}{S_p(\mathbf{0}_{N-1} || 1)} = 1.
\end{equation}
\end{theorem}

\begin{proof}
It can be trivially shown that
\begin{equation}
S_p(\mathbf{0}_N || 1) = \frac{N}{N+1}.
\end{equation}
Therefore
\begin{equation}
\frac{S_p(\mathbf{0}_N || 1)}{S_p(\mathbf{0}_{N-1} || 1)} = \frac{\frac{N}{N+1}}{\frac{N-1}{N}} \xrightarrow{N \rightarrow +\infty} 1.
\end{equation}
\end{proof}

Having proven that GDS fulfils all objective criteria, we now explore how the order $p$ affects the estimation of sparsity.
This will provide insights on what is the appropriate order of GDS under certain circumstances.
Towards this end, the following important theorem shows
that as $p$ increases, vectors are more difficultly qualified as sparse by GDS of $p$-th order. 

\begin{theorem}
\label{Increasing}
Given a vector $\mathbf{c} = [c_1, c_2, \dots, c_N]$, if $p>q$, then $S_p (\mathbf{c}) < S_q (\mathbf{c})$.
\end{theorem}

\begin{proof}

For proving this theorem we can assume that $c_i \! \! > \! \! 1, \forall i \! \! \in \! \! \{ 1,2,\dots,N \}$. Otherwise, since from Theorem \ref{ScaleTheorem} we have that $S_p (\mathbf{c}) = S_p (\alpha \mathbf{c}), \alpha \! > \! 0$, we can multiply each coefficient $c_i$ by a value $\alpha \! > \! \frac{1}{c_1}$ to obtain 
$\alpha c_i \! > \! 1, \forall i \! \! \in \! \! \{ 1,2,\dots,N \}$. 

We have that 
\begin{multline}
\label{pq_ineq}
\frac{ \sum_{i=1}^{N-1} \sum_{j=i+1}^N (c_j - c_i)^p }{c_N^p} = \sum_{i=1}^{N-1} \sum_{j=i+1}^N \frac{(c_j - c_i)^p}{c_N^p}
\\
= \sum_{i=1}^{N-1} \sum_{j=i+1}^N \left( \frac{c_j-c_i}{c_N} \right)^p < \sum_{i=1}^{N-1} \sum_{j=i+1}^N \left( \frac{c_j-c_i}{c_N} \right)^q
\\
= \sum_{i=1}^{N-1} \sum_{j=i+1}^N \frac{(c_j - c_i)^q}{c_N^q} = \frac{ \sum_{i=1}^{N-1} \sum_{j=i+1}^N (c_j - c_i)^q }{c_N^q},
\end{multline}
where the inequality holds because $\frac{c_j-c_i}{c_N}<1$ and $p>q$.
Moreover, since we assumed that $c_i > 1, \forall i \in \{ 1,2,\dots,N \}$, we also have that 
\begin{equation}
c_{N-1}^p + c_{N-2}^p + \dots + c_1^p > c_{N-1}^q + c_{N-2}^q + \dots + c_1^q,
\end{equation}
which combined with inequality (\ref{pq_ineq}) straightforwardly implies that 
\begin{equation}
\frac{1}{N} \frac{ \sum_{i=1}^{N-1} \sum_{j=i+1}^N (c_j - c_i)^p }{c_N^p + (c_{N-1}^p + \dots + c_1^p)} < \frac{1}{N} \frac{ \sum_{i=1}^{N-1} \sum_{j=i+1}^N (c_j - c_i)^q }{c_N^q + (c_{N-1}^q + \dots + c_1^q)},
\end{equation}
which is equivalent to $S_p(\mathbf{c}) < S_q(\mathbf{c})$,
completing the proof.
\end{proof}

In the same direction, we also have the next result, which examines the limit behaviour of $S_p$, when $p$ tends to infinity.

\begin{theorem}
\label{Limit}
For an arbitrary vector $\mathbf{c}$, we have that $\lim_{p \to +\infty} S_p(\mathbf{c}) = 0$.
\end{theorem}

\begin{proof}
For every $p$ we have that 
\begin{multline}
\label{squeeze}
0 \le S_p(\mathbf{c}) = \frac{1}{N} \frac{\sum_{i=1}^{N-1} \sum_{j=i+1}^N (c_j - c_i)^p}{\sum_{i=1}^N c_i^p} 
\\
\le \frac{1}{N} \frac{\sum_{i=1}^{N-1} \sum_{j=i+1}^N (c_j - c_i)^p}{c_N^p} = \frac{1}{N} \sum_{i=1}^{N-1} \sum_{j=i+1}^N \frac{(c_j-c_i)^p}{c_N^p} 
\\
= \frac{1}{N} \sum_{i=1}^{N-1} \sum_{j=i+1}^N \left( \frac{c_j-c_i}{c_N} \right)^p.
\end{multline}
But, since $c_j-c_i < c_N, \forall i,j$, we have that 
\begin{equation}
\lim_{p \to +\infty} \left( \frac{c_j-c_i}{c_N} \right)^p = 0, \forall i,j,
\end{equation} 
which implies that 
\begin{equation}
\lim_{p \to +\infty} \sum_{i=1}^{N-1} \sum_{j=i+1}^N \left( \frac{c_j-c_i}{c_N} \right)^p = 0.
\end{equation} 
Hence, applying the squeeze theorem in inequality (\ref{squeeze}), $\lim_{p \to +\infty} S_p(\mathbf{c}) = 0$.
\end{proof}

\noindent 
In a few words, Theorem \ref{Limit} states that when the order of GDS tends to infinity, all vectors regardless of the values of their coefficients are considered totally non-sparse.

The findings of Theorems \ref{Increasing} and \ref{Limit} actually show that the order $p$ determines the strictness of GDS in qualifying an arbitrary vector as sparse.
More specifically, higher order means more strict GDS.  
This feature offers the appropriate granularity to GDS and allows it to adjust to certain circumstances stemming from the nature of the data and the specific problem to be solved. 
For instance, 
it is anticipated that for data containing few zeros, which are supposed to be inherently non-sparse, a large order might be needed to discriminate among different levels of sparsity. On the contrary, for data with plenty of zeros, smaller orders, i.e., less strict metrics might 
prove to be optimal.
Therefore, 
as we will see in Section \ref{Experiments}, finding and adopting the appropriate GDS metric to a sparsity maximisation problem 
might lead to improved reconstruction results.

\section{Connection between GDS and GI}
\label{Gini}

In this section, we prove that 
GDS of first order collapses to GI \cite{hurley2009comparing,gini1921measurement}. But first of all, let us provide the definition of GI.

\begin{definition}
The Gini Index (GI) of a vector $\mathbf{c} \in \mathbb{R}^N$ is defined as:

\begin{equation*}
GI(\mathbf{c}) = 1 - 2 \sum_{i=1}^N \frac{c_i}{\| \mathbf{c} \|_1} \left( \frac{N - i + \frac{1}{2}}{N} \right) 
\end{equation*}
\end{definition}

\begin{theorem}
GDS of order $1$ is equivalent to GI.
\end{theorem}

\begin{proof}
\begin{multline*}
S_1 (\mathbf{c})
= \frac{1}{N \| \mathbf{c} \|_1} \sum_{i=1}^{N-1} \sum_{j=i+1}^N (c_j - c_i) 
\\
= \frac{1}{N \| \mathbf{c} \|_1} \left( \sum_{i=1}^{N-1} \sum_{j=i+1}^N c_j - \sum_{i=1}^{N-1} \sum_{j=i+1}^N c_i \right) 
\\
= \frac{1}{N \| \mathbf{c} \|_1} \left( \sum_{i=1}^N (i-1) c_i - \sum_{i=1}^N (N-i) c_i \right) 
\\
= \frac{1}{N \| \mathbf{c} \|_1} \left( - N \sum_{i=1}^N c_i - \sum_{i=1}^N c_i + 2 \sum_{i=1}^N i c_i \right) 
\\
= \frac{1}{\| \mathbf{c} \|_1} \left( \sum_{i=1}^N c_i - 2 \sum_{i=1}^N c_i + \frac{2}{N} \sum_{i=1}^N i c_i - \frac{1}{N} \sum_{i=1}^N c_i \right) 
\\
= 1 - \frac{2}{\| \mathbf{c} \|_1} \left( \sum_{i=1}^N c_i - \frac{1}{N} \sum_{i=1}^N i c_i + \frac{1}{2N} \sum_{i=1}^N c_i \right) 
\\
1 - \frac{2}{\| \mathbf{c} \|_1} \sum_{i=1}^N c_i \left( 1 - \frac{i}{N} + \frac{1}{2N} \right) = GI(\mathbf{c}).
\end{multline*}
\end{proof}

The encapsulation of GI within GDS in conjunction with the fact that the former has been proven to be a state-of-the-art metric of sparsity in the literature, demonstrate the power of GDS as a generalised framework for unifying already existing metrics as well as its potential as a framework to develop novel state-of-the-art metrics of sparsity.

\section{A computationally more efficient formula of GDS for high dimensional data}
\label{Efficient}

Although the formula used for the definition of the GDS metric is simple and easy to comprehend, in certain cases, i.e., when the number of vector dimensions is large, it is difficult to compute. A more tractable and computationally efficient formula for GDS of order $p$, with $p$ integer, is presented in this section. Moreover a computational analysis is also provided in order to compare the two formulas. Due to computational reasons, the alternative formula is different for even and odd values of $p$. For this purpose, the two cases are separately presented. The rigorous derivation of these formulas from the original one are provided in the Appendix.

\subsection{Even formula}

\begin{multline}
\label{EvenFormula}
S_{2k} (\mathbf{c}) = 1 + \frac{1}{N\| \mathbf{c} \|_{2k}^{2k}} \Bigg[ \sum_{\omega=1}^{k-1} (-1)^{\omega} 
\binom{2k}{\omega} \| \mathbf{c} \|_{\omega}^{\omega} \| \mathbf{c} \|_{2k - \omega}^{2k - \omega}
\\ 
+ \frac{(-1)^k}{2} \binom{2k}{k} \| \mathbf{c} \|_k^{2k} \Bigg].
\end{multline}

\subsection{Odd formula}

\begin{equation}
\label{OddFormula}
S_{2k+1}(\mathbf{c}) = \frac{1}{N\| \mathbf{c} \|_{2k+1}^{2k+1}} \gamma \,,
\end{equation}
where 
\begin{equation}
\gamma =  \sum_{\omega = 0}^k (-1)^{\omega} \binom{2k+1}{\omega} \sum_{i=1}^N \left( c_i^{2k+1- \omega} f_{\omega}(i) - c_i^{\omega} f_{2k+1- \omega}(i) \right)
\end{equation}
and
\begin{equation}
f_{\omega}(i) = \sum_{j=1}^i c_j^{\omega}.
\end{equation}

\subsection{Computational analysis}

It can be easily proven that the original formula for calculating the sparsity of an $N$-length vector using $S_p$ (see eq. \ref{GDS_def}) requires $(p-1)N\frac{N+1}{2}$ multiplications and $N^2 - 2$ additions, which is in total on the order of $O(N^2 p)$. 
This computational load is for many practical reasons inefficient when $N$ is large. 
The corresponding load when using the ``even'' formula for $S_{2k}$ (see eq. \ref{EvenFormula}), where $p = 2k$, consists of 
$k^2(2N+1) - kN - k$
multiplications and 
$2kN - k + 2$
additions, 
and when using the ``odd'' formula for $S_{2k+1}$ (see eq. \ref{OddFormula}), where now $p = 2k + 1$, of $4k^2N + 6kN + k^2 - k + N - 2$ multiplications and $(2k+5)N - k - 4$ additions.
For both even and odd formulas, the computational complexity is in total $O(k^2 N)$, which is clearly more efficient when approximately $N > p/4$.
However, in the opposite case, the original formula is more tractable.
Indicatively, for a $10000$-dimensional vector, 
the even formula needs around $5000$ times less calculations than the original one for $p=2$ and $100$ times less calculations for $p=100$.
The corresponding numbers for the odd formula are $1700$ and $50$ for $p=1$ and $p=99$, respectively.
Summarising the above, for practical reasons, the new formula proves to be very useful in both even and odd cases.
However, notice how the situation is reversed when $N$ is very small in relation to $p$. For example, for $N=10$ and $p=100$,
the even and odd formulas need respectively around $10$ and $20$ times more calculations than the original formula.
Therefore, both the original and the alternative formulas prove to be important and may be preferred according to the specific conditions.

\section{Differential Sparsity of Normalised Data}
\label{NormalGDS}

In the above analysis, we implicitly assumed that the vector coefficients are commensurate in the sense that they are all measured in the same scale or in some way, they have been normalised, so that they are comparable to each other. Actually, this assumption is indispensable 
for the proposed differential metric to make sense.
However, in practice, this assumption often does not hold. For handling such cases, we propose the normalisation of the coefficients prior to the application of the sparsity metric. The normalisation is accomplished by centralising the data so that they have zero mean and unit standard deviation. A key difference of this approach from the main approach presented in the previous sections is that here we need a dataset of vectors or ideally the underlying distributions, in order to model the mean value and the standard deviation of each of the vector coefficients. So, let 

\begin{equation*}
X = \left( 
\begin{array}{c c c}
x_{1,1} & \hdots & x_{1,n} \\
\vdots & \ddots & \vdots \\
x_{N,1} & \hdots & x_{N,n}
\end{array}
\right)
\end{equation*}
be a dataset consisting of $n$ $N$-dimensional column-vectors.
Let also
\begin{equation*}
\mu_i = \frac{1}{n} \sum_{k=1}^n x_{i,k}, \, i \in \{ 1, 2, \dots, N \}
\end{equation*}
be the mean value and 
\begin{equation*}
\sigma_i^2 = \frac{1}{n-1} \sum_{k=1}^n (x_{i,k} - \mu_i)^2, \, i \in \{ 1, 2, \dots, N \}
\end{equation*}
the standard deviation of the $i$-th coefficient along the dataset.
Then, 

\begin{equation*}
\widehat{x}_{i,j} = \frac{x_{i,j} - \mu_i}{\sigma_i}
\end{equation*}
are the centralised data.
We denote the centralised sparsity as:

\begin{equation*}
\widehat{S}_p ( x_{(:,j)} ) = S( \widehat{x}_{(:,j)}),
\end{equation*}
where $\widehat{x}_{(:,j)} = \left[ \widehat{x}_{1,j}, \widehat{x}_{2,j}, \dots, \widehat{x}_{N,j} \right]^T$.

The above data normalisation scheme permits the application of GDS to any type of data, regardless of the underlying distributions of the representation coefficients.

\section{Experiments}
\label{Experiments}

\begin{figure*}[t!]
\centering
\includegraphics[scale=0.4]{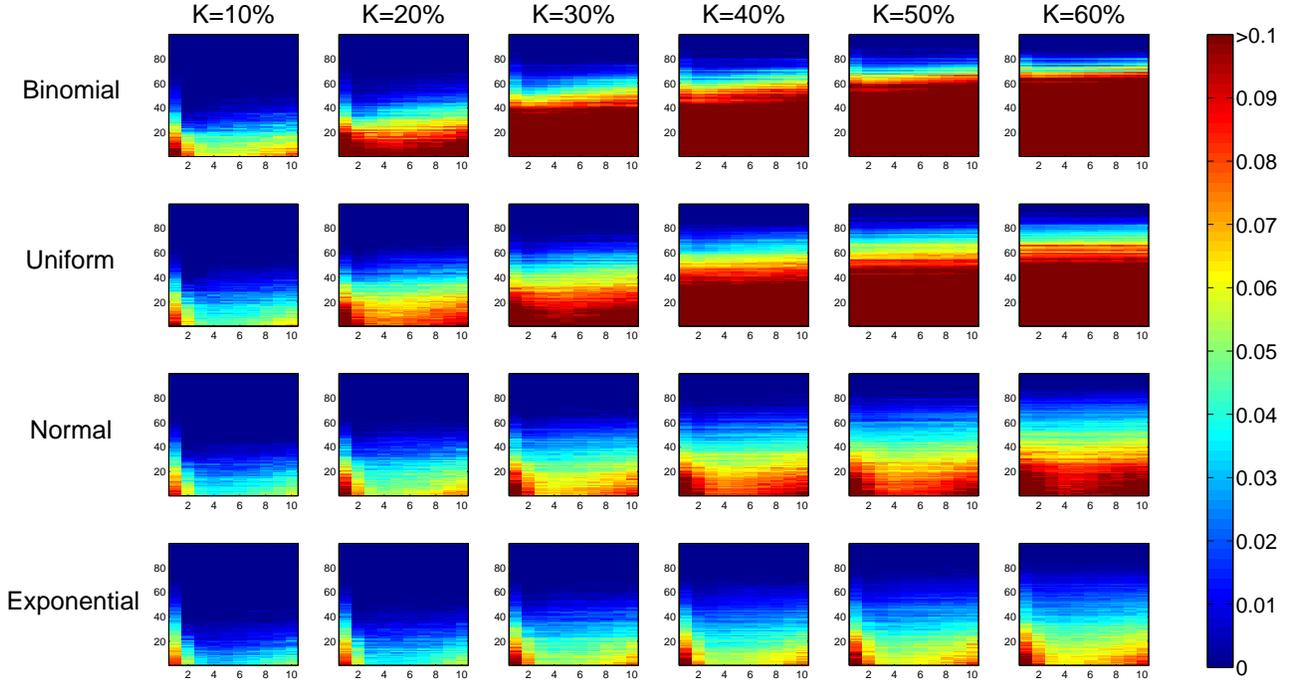}
\caption{Reconstruction error using various values for the order $p$ of GDS and the number $M$ of the reduced dimensions, for different values of the number $K$ of non-zero coefficients of the original data. Horizontal axis: $p$, vertical axis: M, Colourbar: MSE.}
\label{MSEvsMeasurements_maxp_10}  
\end{figure*}

In this Section, we investigate the reconstruction error of randomly projected sparse vectors as a function of the sparsity and the type of the original data, the order of GDS and the number of reduced dimensions.
More specifically, given an original vector $\mathbf{x}_0 \in \mathbb{R}^N$ and an $M \times N$ projection matrix $\mathbf{A}$ with $M<N$, $\mathbf{x}_0$ is projected to $\mathbf{y} = \mathbf{A} \mathbf{x}_0 \in \mathbb{R}^M$, and subsequently reconstructed to the initial space. 
In our work, the entries of $\mathbf{A}$ are generated using 
i.i.d random variables of a zero mean and unit standard deviation Gaussian distribution.
The reconstruction of $\mathbf{x}_0$ from $\mathbf{y}$ is accomplished by solving the following constrained optimisation problem:
\begin{equation}
\label{OptProblem}
\argmax_{\mathbf{x} \in \mathbb{R}^N} S_p(\mathbf{x}) \text{\,\,\, subject to \,\,\,} \mathbf{A} \mathbf{x} = \mathbf{y},
\end{equation}
where $\mathbf{x}$ is the estimate of the original vector and $S_p(\mathbf{x})$ is the sparsity of $\mathbf{x}$.
Essentially, having the prior information that the original vector before compression was sparse, the aim of (\ref{OptProblem}) is to find the sparsest solution in the pool of feasible solutions satisfying the above constraint. For solving this problem, we have employed the iterative Simultaneous Perturbation Stochastic Approximation (SPSA) algorithm \cite{spall1999stochastic}. Actually, 
we adopted the implementation presented in \cite{zonoobi2011gini}, which combined with GI has shown impressive performance in signal reconstruction. 
In our case, the parameters involved in SPSA have been selected based on some previous research results \cite{zonoobi2011gini,sadegh1998optimal}. 

There are two main reasons why we opted to use SPSA. First, it does not make direct reference to the gradient of the objective function. Instead, it approximates the gradient using only two calculations of the objective function per iteration, regardless of the signal dimensionality. This feature renders SPSA computationally very efficient in high-dimensional problems.
Second, it has been proven that SPSA under general conditions converges to the global optima \cite{maryak2001global}.

In our experiment, 
we generated a random vector $\mathbf{x}_0$ of size $N=100$ with $K$ non-zero coefficients. 
We varied $K$ in the range between $10\%$ and $60\%$ and the non-zero coefficients were generated using four different distributions: Binomial, Uniform, Normal and Exponential.
It is worth noting that the smaller the $K$ is, the sparser the vector is, that is $K$ and sparsity are inverse quantities.
For several values of the order $p$ in the range between 1 and 10, we exhaustively varied the number $M$ of the projected dimensions of $\mathbf{y}$ from 1 to 99 and we reconstructed $\mathbf{x}_0$ by employing SPSA/GDS using eq. (\ref{OptProblem}).
Finally, for each setting, we calculated the Mean Square Error (MSE) between the recovered and the original vector.
For ensuring statistical significance, we repeated the whole above approach 100 times and we calculated the average MSE for each triple of values $K$, $M$ and $p$. 

The reconstruction errors that we obtained using the above settings are pooled in Fig. \ref{MSEvsMeasurements_maxp_10}. 
The four rows correspond to the binomial, normal, uniform and exponential data, respectively. 
The subfigures of each row correspond to different values of $K$ (i.e., sparsity level). 
For each subfigure, the horizontal axis depicts the order $p$ of GDS and the vertical axis the number $M$ of projected dimensions. The MSE that corresponds to every pair $(p,M)$ is indicated by the colourbar on the right. 

From Fig. \ref{MSEvsMeasurements_maxp_10}, performing a row-wise (i.e., sparsity oriented) comparison, it is interesting to observe that regardless of the data type, the greater $K$ is, i.e., the less the sparsity of the original vector is, the larger the MSE becomes in general. This behaviour clearly verifies the importance of sparsity in signal reconstruction.
Similarly,
performing a column-wise (i.e., data type oriented) comparison, it is clear that regardless of $K$, the reconstruction error decreases as we move from top (binomial data) to bottom (exponential data).
This can be attributed to the fact that for a specific $K$, although in all four cases we use equal number of non-zeros, in fact GDS tends to consider more sparse those vectors whose sorted absolute coefficients have larger differences. 
Fig. \ref{DataTypes} illustrates the general form of an arbitrary vector generated using either of the above four distributions in the case where $K=40\%$. 
Indicatively, the first order GDS sparsities of these prototypic vectors are approximately $0.60$, $0.73$, $0.76$ and $0.79$, respectively.
Apparently from the above, in terms of GDS, exponential distribution gives the sparsest vectors and thus the most eligible for reconstruction using the adopted methodology.

\begin{figure*}[t!]
\centering
\includegraphics[scale=0.3]{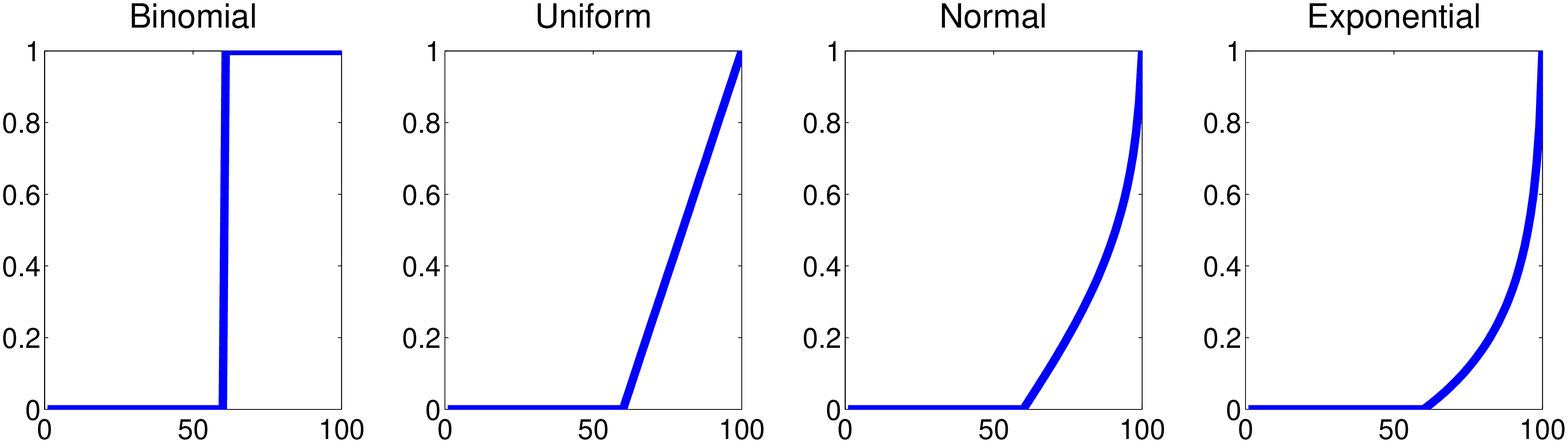}
\caption{General form of the absolute coefficients of vectors generated using binomial, uniform, normal and exponential distributions with $K = 40\%$ non-zeros, sorted in ascending order.}
\label{DataTypes}
\end{figure*}

\begin{figure*}[t!]
\centering
\includegraphics[scale=0.3]{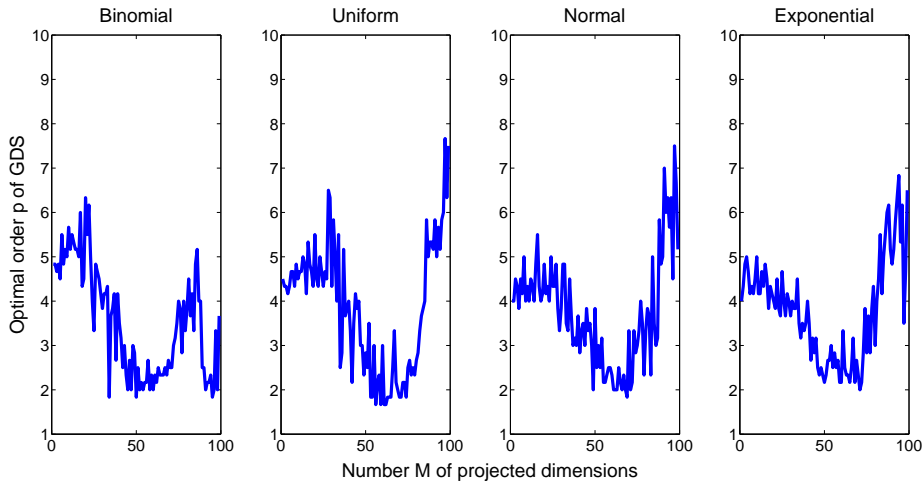}
\caption{Optimal order $p$ of GDS as a function of the number $M$ of projected dimensions.}
\label{Optimal_p_vs_M}
\end{figure*}

Having a closer inspection at each subfigure, we observe that regardless of both $K$ and data type, the reconstruction error has a similar form. More specifically, it is clear that in almost all cases
-- except for binomial and uniform with $K > 40\%$ --
values of $p$ in the range between $2$ and $7$ provide the best results, and this becomes more evident for small values of $M$. 
In this direction, our next concern was to quantify how the optimal $p$ varies as a function of $M$ and $K$ and Fig. \ref{Optimal_p_vs_M} and \ref{Optimal_p_vs_K} serve exactly this purpose.  
In Fig. \ref{Optimal_p_vs_M}, the horizontal axis depicts $M$, while the vertical axis contains the mean optimal $p$, as this has been calculated across the different values of $K$.
From this figure, it is clear that for small $M$ (approximately $<40$), the best reconstruction is obtained by setting $p$ between $4$ and $6$.
Moreover, it is worth noticing that after a small reduction of the mean optimal $p$ for intermediate values of $M$ (i.e., in the interval $40-80$), the superiority of high orders becomes more intensely evident for large values of $M$ (i.e., $>80$). However, it must be pointed out that for large $M$, the difference of reconstruction error among the several values of $p$ becomes negligible as the MSE becomes almost zero.
Finally, it is worth noticing that GI, which recall that is obtained for $p=1$, is never the optimal choice in reconstruction, justifying the use of higher orders and proving the superiority of GDS. 
In summary, the above findings explicitly demonstrate how GDS of high orders reduces the least number of projected dimensions required in order to perform almost perfect reconstruction of sparse signals.

 \begin{figure*}[t!]
\centering
\includegraphics[scale=0.3]{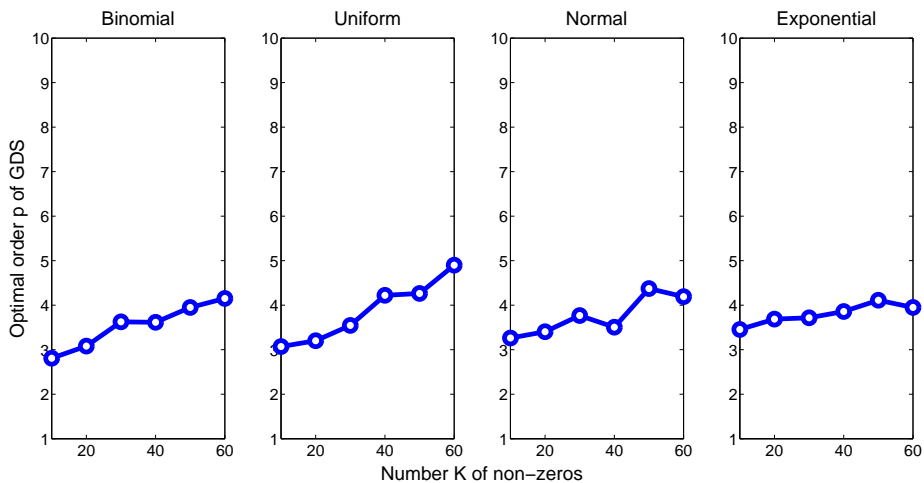}
\caption{Optimal order $p$ of GDS as a function of the number $K$ of non-zero coefficients.}
\label{Optimal_p_vs_K}
\end{figure*}

Similar is the case when we investigate the optimal $p$ as a function of the number $K$ of non-zero coefficients. 
In Fig. \ref{Optimal_p_vs_K}, the horizontal axis depicts $K$, while in the vertical axis is the mean optimal $p$, as this has been calculated this time across the different values of $M$. 
Again, it is interesting to observe that GI never offers the best reconstruction performance. 
Instead, for every $K$, orders of GDS larger than $3$ are needed.
In particular, as the sparsity of the data decreases, 
larger values of the order $p$ are required for better reconstructing a signal and on average $p=4$ provides the best results. 
This outcome can be attributed to the strictness that $p$ provides to GDS (cf. Theorems \ref{Increasing} and \ref{Limit}, Section \ref{GDSp}) and explicitly demonstrates how GDS loosens the bounds of the assumed sparsity of the original data offering more capacity in reconstructing lowly sparse signals.
Summarising the above results, 
GDS can undoubtedly substitute GI in sparse signal reconstruction.
This inference in conjunction with the proven prevalence of GI among other top performing 
sparsity metrics 
\cite{zonoobi2011gini} induces the superiority of GDS over the state-of-the-art.

Closing this section, it is important to stress that the previous experimental study offers a rule of thumb in deciding what is the optimal order $p$ of GDS in a certain compression-reconstruction scenario.
In this sense, Fig. \ref{MSEvsMeasurements_maxp_10} actually serves as a look-up table indicating the appropriate $p$ depending on the type and sparsity of the original data as well as the target compression level. The following comprises a concrete example demonstrating how the above rule of thumb could work. 
Consider that we have an $100$-dimensional signal containing $20$ zeros and whose non-zero coefficients have been generated by either the same or different normal distributions. Note that in the latter case, the coefficients should be normalised using the approach presented in Section \ref{NormalGDS}. Also consider that we would like to compress this signal to $10$ dimensions. Based on these settings, we refer to the subfigure lying on the third row (i.e., Normal), second column (i.e., $K = 20\%$) of Fig. \ref{MSEvsMeasurements_maxp_10}. 
In this subfigure, taking a horizontal cross section at value $10$ of $y$-axis (i.e., the target reduced dimension), we can find the least reconstruction error using the colourbar on the right. Obviously, this error corresponds to the optimal $p$, for the above settings, which in this specific example is approximately $4$.

\section{Conclusions}
\label{Conclusion}

The main contribution of this paper is a Generalised Differential Sparsity (GDS) framework for generating novel sparsity metrics.
The proposed framework accumulates a number of advantages.
First, it is characterised by the flexibility to generate metrics well-tailored to specific problem and data requirements.
Second, it has been shown that, in contrast to other sparsity metrics, GDS satisfies a set of benchmark criteria proving its credibility as an effective metric for measuring signal sparsity.
Third, it has been proven that GI constitutes a specific case of GDS demonstrating the generalisation power of the latter to unify already existing metrics.
Fourth, it can be calculated using alternative formulas with complementary computational advantages, therefore allowing for its efficient calculation under different settings.

The above features offer GDS a great potential as a general purpose framework regardless of the domain it is used.
As a matter of fact, in this paper, the above potential has been demonstrated within the context of Compressive Sampling (CS) in the process of reconstructing signals heavily compressed using random projections.
Along these lines, through an extensive experimental study on synthetic data whose coefficients are generated using binomial, uniform, normal and exponential distributions, 
GDS has proven to be more effective than GI in measuring sparsity, in terms of signal reconstruction capability.
More specifically, GDS in comparison to GI, loosens the assumptions of both the least number of projected dimensions and the inherent sparsity of the original data, required in order to almost perfectly reconstruct a compressed signal. 
The superiority of GDS against GI in conjunction with the fact that the latter has categorically outperformed other state-of-the-art sparsity metrics in signal reconstruction \cite{zonoobi2011gini}, places GDS in the pole position in the sparsity metric literature.

In the near future, we plan to replicate the experimental study, presented in this paper, on real data instead of synthetic. This will extend the impact of GDS to more types of data. Since often the coefficients of real data can be approximated and modelled by the four distributions used in this paper, we envisage that the results will be similar to the ones presented here.
Moreover, as GI occupies a leading position in the reconstruction of signals 
contaminated with noise \cite{zonoobi2011gini},
we also intend to investigate the effectiveness of GDS in the same problem.
Towards this direction, we anticipate that the privilege of GDS to contain an adjusting parameter has the potential to offer the appropriate robustness in the view of noise. 
Finally, although signal reconstruction is an extended field, adapting GDS in other case studies as well will reinforce its potential. To this end, based on the proven contribution of sparsity to optimisation problems, we plan to investigate how could GDS be applied in multi-objective optimisation, which is an emerging field with many applications.

\appendix
\label{Appendix}

\subsection{Proof of $P_3$ criterion}

\begin{theorem}
\label{RobinHoodTheorem}
GDS satisfies $P_3$: Robin Hood, i.e.:
\begin{equation}
\label{RobinHood}
S_p \left( \left[ c_1, \dots, c_i+\alpha, \dots, c_j-\alpha, \dots, c_N \right] \right) < S_p \left( \mathbf{c} \right),
\end{equation} 
for all $\alpha, c_i, c_j$, such that $c_j>c_i$ and $0<\alpha<\frac{c_j-c_i}{2}$.
\end{theorem}

\begin{table*}
\caption{Correspondence between the indices $c_i$ before and $d_i$ after a Robin Hood operation.}
\label{RobinHoodIndices}
\begin{equation*} \label{Indices}
\left( \begin{array}{c c c c c c c c c c c c c c c c c}
d_1 & \cdots & d_{i-1} & d_i      & \cdots & d_{i+m-1} & d_{i+m} & d_{i+m+1} & \cdots & d_{j-n-1} & d_{j-n} & d_{j-n+1} & \cdots & d_j & d_{j+1} & \cdots & d_N \\
\\
c_1 & \cdots & c_{i-1} & c_{i+1} & \cdots & c_{i+m} & c_i + \alpha & c_{i+m+1} & \cdots & c_{j-n-1} & c_j - \alpha & c_{j-n} & \cdots & c_{j-1} & c_{j+1} & \cdots & c_N \\
\end{array} \right)\,
\end{equation*}
\end{table*}

\begin{proof}

Let $d_1 \le d_2 \le \dots \le d_N$ be the sorted coefficients obtained after a Robin Hood operation on $\mathbf{c}$. Then we have:
\begin{equation}
d_k = \left\{ \begin{array}{l l}
c_k, & \, \, 1 \le k \le i-1\\
c_{k+1} & i \le k \le i + m - 1\\
c_i + \alpha & k = i + m\\
c_k & i + m + 1 \le k \le j - n - 1\\
c_j - \alpha & k = j - n\\
c_{k-1} & j - n + 1 \le k \le j\\
c_k & j + 1 \le k \le N
 \end{array}
 \right. \,.
\end{equation}
For further clarity see also Table \ref{RobinHoodIndices}.
Clearly proving eq. (\ref{RobinHood}) is equivalent to proving that
\begin{equation}
\label{PartialS}
\frac{\partial S_p \left( \left[ c_1, \dots, c_i-\alpha, \dots, c_j+\alpha, \dots, c_N \right] \right)}{\partial \alpha} < 0.
\end{equation} 
Expanding $S_p \left( \left[ c_1, \dots, c_i-\alpha, \dots, c_j+\alpha, \dots, c_N \right] \right)$ and eliminating all terms that do not contain $\alpha$, inequality (\ref{PartialS}) reduces to 
\begin{equation}
\frac{\partial F(\alpha)}{\partial \alpha} < 0,
\end{equation} 
where
\begin{multline*}
F(\alpha) = \sum_{k=1}^{i+m-1} \! \! \! \! \left( d_{i+m} - d_k \right)^p + \! \! \! \! \! \sum_{k=i+m+1}^N \! \! \! \! \! \! \! \left( d_k - d_{i+m} \right)^p
\\ 
+ \sum_{k=1}^{j-n-1} \! \! \! \! \left( d_{j-n} - d_k \right)^p + \! \! \! \! \! \sum_{k=j-n+1}^N \! \! \! \! \! \! \! \left( d_k - d_{j-n} \right)^p 
\\
= \sum_{\substack{k=1 \\ k \ne i}}^{i+m} \! \left[ ( c_i+\alpha ) - c_k \right]^p + \! \! \! \! \! \sum_{k=i+m+1}^{j-n-1} \! \! \! \! \! \! \! \left[ c_k - (c_i+\alpha) \right]^p
\\ 
+ \left[ (c_j - \alpha) - (c_i + \alpha) \right]^p 
\\ 
+ \sum_{\substack{k=j-n \\ k \ne j}}^N \! \! \! \! \left[ c_k-(c_i+\alpha) \right]^p + \sum_{\substack{k=1 \\ k \ne i}}^{i+m} \! \left[ (c_j- \alpha) - c_k \right]^p
\\ 
+ \left[ (c_j- \alpha) - (c_i + \alpha) \right]^p
\\ 
+ \sum_{k=i+m+1}^{j-n-1} \! \! \! \! \! \! \! \left[ (c_j - \alpha) - c_k \right]^p + \! \! \! \sum_{\substack{k=j-n \\ k \ne j}}^N \! \! \! \! \left[ c_k - (c_j - \alpha) \right]^p
\\
= \sum_{\substack{k=1 \\ k \ne i}}^{i+m} \! \Bigg\{ \left[ ( c_i+\alpha ) - c_k \right]^p + \left[ (c_j- \alpha) - c_k \right]^p \Bigg\} 
\\
+ \sum_{k=i+m+1}^{j-n-1} \! \Bigg\{ \left[ c_k - (c_i+\alpha) \right]^p + \left[ (c_j - \alpha) - c_k \right]^p \Bigg\}
\\
+ \sum_{\substack{k=j-n \\ k \ne j}}^N \! \Bigg\{ \left[ c_k-(c_i+\alpha) \right]^p + \left[ c_k - (c_j - \alpha) \right]^p \Bigg\} 
\\
+ 2 \left[ (c_j - \alpha) - (c_i + \alpha) \right]^p.
\end{multline*}
Hence,
\begin{multline*}
\frac{\partial F}{\partial \alpha} = 
p \sum_{\substack{k=1 \\ k \ne i}}^{i+m} \Bigg\{ \left[ (c_i + \alpha) - c_k \right]^{p-1} - \left[ (c_j - \alpha) - c_k \right]^{p-1} \Bigg\} 
\\
-p \sum_{k=i+m+1}^{j-n-1} \Bigg\{ \left[ c_k - (c_i + \alpha) \right]^{p-1} + \left[ (c_j - \alpha) - c_k \right]^{p-1} \Bigg\}
\\ 
-p \sum_{\substack{k=j-n \\ k \ne j}}^N \Bigg\{ \left[ c_k - (c_i + \alpha) \right]^{p-1} - \left[ c_k - (c_j - \alpha) \right]^{p-1} \Bigg\}
\\
-4p \left[ c_j - c_i -2\alpha \right]^{p-1}.
\\
\end{multline*}

But 
\begin{multline}
\label{D1_Assumption}
0 < \alpha < \frac{c_j - c_i}{2} \Leftrightarrow 0 < 2 \alpha < c_j - c_i
\\ 
\Leftrightarrow c_j - \alpha > c_i + \alpha \Leftrightarrow (c_j - \alpha) - c_k > (c_i + \alpha) - c_k
\end{multline}
For $k \in \{ 1, 2, \dots, i+m \}$, we have:
\begin{multline*}
(c_j - \alpha) - c_k > (c_i + \alpha) - c_k > 0 \Rightarrow 
\\
p \sum_{\substack{k=1 \\ k \ne i}}^{i+m} \Bigg\{ \left[ (c_i + \alpha) - c_k \right]^{p-1} - \left[ (c_j - \alpha) - c_k \right]^{p-1} \Bigg\} < 0.
\end{multline*}
Similarly, for $k \in \{ j-n, j-n+1, \dots, N \}$:
\begin{multline*}
0 > (c_j - \alpha) - c_k > (c_i + \alpha) - c_k \Rightarrow
\\
0 < c_k - (c_j - \alpha) < c_k - (c_i + \alpha) \Rightarrow
\\
-p \sum_{\substack{k=j-n \\ k \ne j}}^N \Bigg\{ \left[ c_k - (c_i + \alpha) \right]^{p-1} - \left[ c_k - (c_j - \alpha) \right]^{p-1} \Bigg\} < 0.
\end{multline*}
Moreover, for $k \in \{ i+m+1, j-n-1, \dots, N \}$, both $c_k - (c_i + \alpha)$ and $(c_j - \alpha) - c_k$ are positive, which implies that 
\begin{equation}
-p \sum_{k=i+m+1}^{j-n-1} \Bigg\{ \left[ c_k - (c_i + \alpha) \right]^{p-1} + \left[ (c_j - \alpha) - c_k \right]^{p-1} \Bigg\} > 0.
\end{equation}
Finally, obviously 
\begin{equation}
-4p \left[ c_j - c_i -2\alpha \right]^{p-1} < 0.
\end{equation}
Hence, all terms of $\frac{\partial F}{\partial \alpha}$ are negative, therefore $\frac{\partial F}{\partial \alpha} < 0$, which 
completes the proof.

\end{proof}

\subsection{A computationally more efficient GDS formula for even values of $p$}

\begin{theorem}
\begin{multline*}
S_{2k} (\mathbf{c}) = 1 + \frac{1}{N\| \mathbf{c} \|_{2k}^{2k}} \Bigg[ \sum_{\omega=1}^{k-1} (-1)^{\omega} 
\binom{2k}{\omega} \| \mathbf{c} \|_{\omega}^{\omega} \| \mathbf{c} \|_{2k - \omega}^{2k - \omega}
\\ 
+ \frac{(-1)^k}{2} \binom{2k}{k} \| \mathbf{c} \|_k^{2k} \Bigg]
\end{multline*}
\end{theorem}

\begin{proof}
From the definition of GDS (see eq. \ref{GDS_definition}), we have:
\begin{equation*}
S_{2k} (\mathbf{c}) = \frac{1}{N\| \mathbf{c} \|_{2k}^{2k}} \alpha, 
\end{equation*}
where
\begin{multline*}
\alpha = \sum_{i=1}^{N-1} \! \sum_{j=i+1}^N (c_j - c_i)^{2k} \\ 
= \sum_{i=1}^{N-1} \sum_{j=i+1}^N \sum_{\omega = 0}^{2k} (-1)^{\omega} \binom{2k}{\omega} c_j^{\omega} c_i^{2k - \omega} \\
= \sum_{i=1}^{N-1} \sum_{j=i+1}^N \Bigg[ (c_j^{2k} + c_i^{2k}) + \sum_{\omega=1}^{k-1} (-1)^{\omega} \binom{2k}{\omega} c_j^{\omega} c_i^{2k - \omega} 
\\
+ (-1)^k \binom{2k}{k} c_j^k c_i^k + \sum_{\omega = k + 1}^{2k - 1} (-1)^{\omega} \binom{2k}{\omega} c_j^{\omega} c_i^{2k - \omega} \Bigg] \\
= \sum_{i=1}^{N-1} \sum_{j=i+1}^N c_j^{2k} + \sum_{i=1}^{N-1} \sum_{j=i+1}^N c_i^{2k}
\\ 
+ \sum_{i=1}^{N-1} \sum_{j=i+1}^N \sum_{\omega = 1}^{k-1} \Bigg[ (-1)^{\omega} \binom{2k}{\omega} (c_j^{\omega} c_i^{2k - \omega} + c_j^{2k - \omega} c_i^{\omega}) \Bigg] 
\\
+ \sum_{i=1}^{N-1} \sum_{j=i+1}^N \Bigg[ (-1)^k \binom{2k}{k} c_j^k c_i^k \Bigg]
\\
= \sum_{i=1}^N (i-1) c_i^{2k} + \sum_{i=1}^N (N-i) c_i^{2k} + \sum_{\omega = 1}^{k-1} \Bigg\{ (-1)^{\omega} \binom{2k}{\omega} 
\\
\Bigg[ \sum_{i=1}^{N-1} \sum_{j=i+1}^N (c_j^{\omega} c_i^{2k - 
\omega} + c_j^{2k - \omega} c_i^{\omega} ) + \sum_{i=1}^N c_i^{2k} - \sum_{i=1}^N c_i^{2k} \Bigg] \Bigg\} 
\\
+ \frac{(-1)^k}{2} \binom{2k}{k} \Bigg[ 2 \sum_{i=1}^{N-1} \sum_{j=i+1}^N c_j^k c_i^k + \sum_{i=1}^N 
c_i^{2k} - \sum_{i=1}^N c_i^{2k} \Bigg]
\\
= (N-1) \sum_{i=1}^N c_i^{2k} + \sum_{\omega = 1}^{k-1} \Bigg\{ (-1)^{\omega} \binom{2k}{\omega} 
\\
\Bigg[ \sum_{i=1}^N \sum_{j=1}^N c_j^{\omega} c_i^{2k - \omega} - \sum_{i=1}^N c_i^{2k} \Bigg] \Bigg\} 
\\
+ \frac{(-1)^k}{2} \binom{2k}{k} \Bigg[ \sum_{i=1}^N \sum_{j=1}^N c_j^k c_i^k - \sum_{i=1}^N c_i^{2k} 
\Bigg] 
\\
= N \sum_{i=1}^N c_i^{2k} + \sum_{\omega = 1}^{k-1} (-1)^{\omega} \binom{2k}{\omega} 
\Bigg[ \sum_{i=1}^N c_i^{2k - \omega} \sum_{j=1}^N c_j^{\omega} \Bigg]
\\ 
+ \frac{(-1)^k}{2} \binom{2k}{k} \Bigg[ \sum_{i=1}^N c_i^k \sum_{j=1}^N c_j^k \Bigg] 
\\
- \Bigg[ 1 + \sum_{\omega=1}^{k-1} (-1)^{\omega} 
\binom{2k}{\omega} - \frac{(-1)^k}{2} \binom{2k}{k} \Bigg] \sum_{i=1}^N c_i^{2k} 
\\
= \! \! N \| \mathbf{c} \|_{2k}^{2k} 
+ \sum_{\omega = 1}^{k-1} (-1)^{\omega} \! \binom{2k}{\omega} \! \| \mathbf{c} \|_{\omega}^{\omega} \| \mathbf{c} \|
_{2k - \omega}^{2k - \omega} + \frac{(-1)^k}{2} \! \binom{2k}{k} \! \| \mathbf{c} \|_k^{2k},
\end{multline*}
since 
\begin{multline*}
1 + \sum_{\omega = 1}^{k-1} (-1)^{\omega} \binom{2k}{\omega} - \frac{(-1)^k}{2} \binom{2k}{k} 
\\
= \sum_{\omega = 0}^k (-1)^{\omega} \binom{2k}{\omega} + \frac{(-1)^k}{2} \binom{2k}{k} = 0.
\end{multline*}
Notice that in the above proof, we made use of the identity $\binom{2k}{\omega} = \binom{2k}{2k - \omega}$.
Therefore, 
\begin{multline*}
S_{2k} (\mathbf{c}) = \frac{1}{N\| \mathbf{c} \|_{2k}^{2k}} \alpha \\
= 1 + \frac{1}{N\| \mathbf{c} \|_{2k}^{2k}} \Bigg[ \sum_{\omega=1}^{k-1} (-1)^{\omega} 
\binom{2k}{\omega} \| \mathbf{c} \|_{\omega}^{\omega} \| \mathbf{c} \|_{2k - \omega}^{2k - \omega}
\\ 
+ \frac{(-1)^k}{2} \binom{2k}{k} \| \mathbf{c} \|_k^{2k} \Bigg] \\
\end{multline*}
\end{proof}

\subsection{A computationally more efficient GDS formula for odd values of $p$}

\begin{theorem}
\begin{equation*}
S_{2k+1}(\mathbf{c}) = \frac{1}{N\| \mathbf{c} \|_{2k+1}^{2k+1}} \gamma 
\end{equation*}
where 
\begin{equation}
\gamma =  \sum_{\omega = 0}^k (-1)^{\omega} \binom{2k+1}{\omega} \sum_{i=1}^N \left( c_i^{2k+1- \omega} f_{\omega}(i) - c_i^{\omega} f_{2k+1- \omega}(i) \right)
\end{equation}
and
\begin{equation}
f_{\omega}(i) = \sum_{j=1}^i c_j^{\omega}
\end{equation}
\end{theorem}

\begin{proof}
Again, from the definition of GDS (see eq. \ref{GDS_definition}), we have:
\begin{equation}
\label{2k_1}
S_{2k+1} (\mathbf{c}) = \frac{1}{N\| \mathbf{c} \|_{2k+1}^{2k+1}} \beta,
\end{equation}
where
\begin{multline*}
\beta = \sum_{i=1}^{N-1} \! \sum_{j=i+1}^N (c_j - c_i)^{2k+1} \\
= \sum_{i=1}^{N-1} \sum_{j=i+1}^N \sum_{\omega = 0}^{2k+1} (-1)^{\omega} \binom{2k+1}{\omega} c_j^{2k+1- \omega} c_i^{\omega} 
\\
= \sum_{i=1}^{N-1} \sum_{j=i+1}^N \Big[ c_j^{2k+1} - \binom{2k+1}{1} c_j^{2k} c_i + \dots 
\\
- \binom{2k+1}{2} c_j^2 c_i^{2k+1} + \binom{2k+1}{1} c_j c_i^{2k} - c_i^{2k+1} \Big]
\\
= \sum_{i=1}^{N-1} \sum_{j=i+1}^N (c_j^{2k+1} - c_i^{2k+1} )
\\ 
- \binom{2k+1}{1} \sum_{i=1}^{N-1} \sum_{j=i+1}^N (c_j^{2k} c_i - c_j c_i^{2k} )
\\ 
+ \dots + (-1)^k \binom{2k+1}{k} \sum_{i=1}^{N-1} \sum_{j=i+1}^N (c_j^{k+1} c_i^k - c_j^k c_i^{k+1} )
\\
= 2 \sum_{i=1}^N c_i^{2k+1} f_0(i) - (N+1) \| \mathbf{c} \|_{2k+1}^{2k+1}
\\ 
- \binom{2k+1}{1} \sum_{i=1}^N \left( c_i^{2k} f_1(i) - c_i f_{2k}(i) \right)
\\ 
+ \dots + (-1)^k \binom{2k+1}{k} \sum_{i=1}^N \left( c_i^{k+1} f_k(i) - c_i^k f_{k+1}(i) \right),
\end{multline*}
But, 
\begin{multline*}
\sum_{i=1}^N c_i^0 f_{2k+1}(i) = \sum_{i=1}^N \sum_{j=1}^i c_j^{2k+1} = \sum_{i=1}^N (N+1-i) c_i^{2k+1}
\\
= (N+1) \sum_{i=1}^N c_i^{2k+1} - \sum_{i=1}^N i c_i^{2k+1}  
\\
\Leftrightarrow (N+1) \| \mathbf{c} \|_{2k+1}^{2k+1} = \sum_{i=1}^N c_i^0 f_{2k+1}(i) + \sum_{i=1}^N i c_i^{2k+1}
\end{multline*}
Therefore (\ref{2k_1}) becomes:
\begin{multline*}
2 \sum_{i=1}^N c_i^{2k+1} f_0(i) - \sum_{i=1}^N c_i^0 f_{2k+1}(i) - \sum_{i=1}^N i c_i^{2k+1}
\\ 
+ \sum_{\omega = 0}^k (-1)^{\omega} \binom{2k+1}{\omega} \sum_{i=1}^N (c_i^{2k+1- \omega} f_{\omega}(i) - c_i^{\omega} f_{2k+1- \omega}(i)
\\
= \sum_{i=1}^N c_i^{2k+1} f_0(i) - \sum_{i=1}^N i c_i^{2k+1}
\\ 
+ \sum_{\omega = 0}^k \! \Bigg[ (-1)^{\omega} \binom{2k+1}{\omega} \! \sum_{i=1}^N \left( c_i^{2k+1- \omega} f_{\omega}(i) - c_i^{\omega} f_{2k+1- \omega}(i) \right) \! \Bigg]
\\
= \! \sum_{\omega = 0}^k \! \Bigg[ (-1)^{\omega} \! \binom{2k+1}{\omega} \! \sum_{i=1}^N \! \left( c_i^{2k+1- \omega} f_{\omega}(i) - c_i^{\omega} f_{2k+1- \omega}(i) \right) \! \Bigg]
\end{multline*}
\end{proof}

\ifCLASSOPTIONcaptionsoff
  \newpage
\fi

\bibliographystyle{IEEEtran}
\bibliography{references}

\end{document}